\documentclass[twocolumn]{autart}
\usepackage[update,prepend]{epstopdf}
\usepackage{amsmath,rotating}
\usepackage{hyperref}
\usepackage{amsfonts}
\usepackage{amssymb}
\usepackage{times,subfigure}
\usepackage{comment}
\usepackage{epstopdf}
\usepackage{algorithmic}
\usepackage{algorithm}
\usepackage[scanall]{psfrag}
\newtheorem{Theorem}{Theorem}

\newtheorem{Lemma}{Lemma}
\newtheorem{Proposition}{Proposition}

\newtheorem{Remark}{Remark}

\def\sqw{\hfill\hbox{\lower.1ex\hbox{$\sqcup$}
    \kern-1.02em\lower.1ex\hbox{$\sqcap$}}\ }

\DeclareMathOperator*{\argmin}{\arg\min}

\DeclareMathOperator*{\trace}{tr}

\newenvironment{proof}{\vspace{1ex}\noindent{\bf Proof}\hspace{0.5em}}
	{\hfill\qed\vspace{1ex}}

\newcommand{\av}{\mathbf{a}}
\newcommand{\bv}{\mathbf{b}}

\newcommand{\hv}{\mathbf{h}}

\newcommand{\ev}{\mathbf{e}}

\newcommand{\R}{\mathbb{R}}

\newcommand{\SUB}{\mathcal{S}}
\newcommand{\deltabf}{\mathbf{\delta}}

\begin{document}
\begin{frontmatter}

\title{On the Nuclear Norm Heuristic \\for a Hankel Matrix Completion Problem}

\author[First]{Liang Dai}
\author[First]{and Kristiaan Pelckmans}

\address[First]{Division of Systems and Control,
	Department of Information Technology, \\
	Uppsala University, Sweden \\
	e-mail: liang.dai@it.uu.se, kristiaan.pelckmans@it.uu.se.}

\maketitle

\begin{abstract}
    This note addresses the question if and why the nuclear norm heuristic can recover an impulse response generated by a stable single-real-pole system, if elements of the upper-triangle of the associated Hankel matrix are given.
    Since the setting is deterministic, theories based on stochastic assumptions for low-rank matrix recovery do not apply in the considered situation.
    A 'certificate' which guarantees the success of the matrix completion task is constructed by exploring the structural information of the hidden matrix. Experimental results and discussions regarding the nuclear norm heuristic applied to a more general setting are also given.
\end{abstract}
\date{}
\end{frontmatter}

\begin{section}{Introduction}
Techniques of convex relaxation using the nuclear norm heuristic have become increasingly popular in the systems and control community, see e.g. the examples reported in \cite{c8}, \cite{c1} and the discussions therein.
This note provides a theoretical justification for the usage of the nuclear norm heuristic when it is applied to an fundamental task in systems theory, i.e. to recover the impulse response of a system from the first few entries of the related series. Precisely, we make the following assumptions throughout the note:
(1) the provided entries are exact, i.e. there are no noise present,
(2) the first $n$ entries of the impulse response are provided while the last $n-1$ entries are to be completed.

The problem considered can be casted as a special case of the 'matrix completion' problem \cite{c6,c7,c3}. However, in this work, the sampled entries are given deterministically, while the 'matrix completion' problems are typically analyzed using random sampling patterns. Furthermore, the underlying matrix is a structural (Hankel) matrix. These differences make the theories in the literature not applicable directly to this problem.

While this task can be easily solved using standard techniques \cite{c8}, the rationale for this work is that to provide a complete picture for understanding how the nuclear norm heuristic performs on this fundamental problem.

This contribution is organized as follows. The main theorem and its proof are given in Section 2. Section 3 discusses a more general matrix completion problem and concludes the note. Appendix gives the technical proofs for the Facts in Section 2.

The following notational conventions will be used. Vectors are denoted in boldface, scalars are denoted in lowercase, matrices as capital letters, and sets are represented as calligraphic letters.
$\mathcal{H}_n$ denotes the set of $n\times n$ Hankel matrices, $I_n$ denotes the identity matrix of size $n\times n$, $\ev_i$ denotes the unit vector with only the $i$-th element to be one and all the other elements zero,  $\|\cdot\|_{\ast}$ represents the nuclear norm (sum of all the singular values) of a matrix, $\|\cdot\|_{2}$ represents the spectral norm of a matrix, and $\|\cdot\|_{F}$ represents the Frobenius norm of a matrix.
\end{section}

\begin{section}{Results}
The following theorem states the result formally.
\vspace{-3mm}
\begin{Theorem}{\em
	Given $-1<h<1$, define vector $\mathbf{h} \in \mathbb{R}^n$ as $\mathbf{h}=[1,h,h^2, \dots,h^{n-1}]^T$,
	and matrix $ G_0\in \mathcal{H}_n$ as $\mathbf{h}\mathbf{h}^T$.
	Consider the following application of the nuclear norm heuristic:
	\begin{align}\label{eq.opt}
		\hat{G}_0
		\triangleq  &\argmin_{G\in\mathcal{H}_n} \|G\|_{*} \\
		&\mbox{ \ s.t. \ } G(i,j) = G_0(i,j), \forall\  (i+j)\le n+1, \nonumber
	\end{align}
	it holds that $\hat{G}_0$ is unique and $\hat{G}_0 = G_0$.
}\end{Theorem}

\begin{Remark}
	Since matrix $G_0$ is of low rank (of rank one), when
	\begin{align}\label{eq.rank}
		\tilde{G}_0
		\triangleq &\argmin_{G\in\mathcal{H}_n}\  \mbox{rank}(G) \\
		&\mbox{ \ s.t. \ } G(i,j) = G_0(i,j), \forall (i+j)\le n+1, \nonumber
	\end{align}
	is solved, one has that $\tilde{G}_0 = G_0$, see e.g. \cite{c4,c10}.
\end{Remark}

\begin{subsection}{Proof of Theorem 1}
The {\em roadmap} for the proof of the Theorem 1 is as follows:
Lemma 1 gives a sufficient condition for the recovery of $G_0$ by solving eq. (\ref{eq.opt}).
Lemma 2 and Lemma 3 are devoted to build the 'certificate' which guarantees that the sufficient condition in Lemma 1 is satisfied.

Based on the matrices $G_0$ and $G$ in Theorem 1, define:
\begin{equation}
	H = G_0 - G,
	\label{eq.H}	
\end{equation}
Notice that by construction, all the entries of $H$ in the upper triangle part are zero, so $H$ can be decomposed as
	\begin{equation}
		H = \sum_{i=1}^{n-1}v_iG_i,
	\end{equation}
	where $\{G_i\}_{i=1}^{n-1}$ are the basis matrices with the elements of the $i$th lower anti-diagonal equal to 1 and the others equal to zero and $v_i \in \mathbb{R}, \forall i = 1,\cdots,n-1$.
	For instance, when $n=3$, one has that
	\begin{equation}
		H= \left( \begin{array}{ccc}
		0 & 0 & 0 \\
		0 & 0 & v_1 \\
		0 & v_1 & v_2 \end{array} \right)\\
        =
		v_1\left( \begin{array}{ccc}
		0 & 0 & 0 \\
		0 & 0 & 1 \\
		0 & 1 & 0 \end{array} \right)
		+
		v_2\left( \begin{array}{ccc}
		0 & 0 & 0 \\
		0 & 0 & 0 \\
		0 & 0 & 1 \end{array} \right).
		\label{eq.Q}
	\end{equation}

Define the projection matrix $$P = \frac{G_0}{\|\mathbf{h}\|_2^2} $$ and its complement projection matrix as $Q = I_n -P.$

Proposition 1 will be used later, which characterizes the nuclear norm as the dual norm of the spectral norm for a given matrix \cite{c3}.
\begin{Proposition}
	Given $A\in\mathbb{R}^{n\times n}$ matrix, then
	\begin{equation}
		\|A\|_* = \sup\{\trace(MA): \|M\|_2\leq 1, M\in\mathbb{R}^{n\times n}\}.
		\label{dual}
	\end{equation}
\end{Proposition}
\vspace{-2mm}
The following result will be needed in the Lemma 3.
\begin{Proposition}
	Given $H$ as defined in eq. (\ref{eq.H}), if $H \neq 0$, then $QHQ \neq 0$.
\end{Proposition}
\begin{proof}
	We prove that the only possibility for $QHQ = 0$ to hold is when $H = 0$.
    Notice that $H = (P+Q)H(P+Q)$, expanding this equality, we have that
     $$H = PHP + PHQ + QHP + QHQ.$$
	Hence if $QHQ = 0$, we have that
\begin{align*}
H &= PHP + PHQ + QHP\\
  &= PH + QHP.
\end{align*}
    Since $P = \frac{\mathbf{h}\mathbf{h}^{T}}{\|\mathbf{h}\|_2^2}$, the previous relation implies that $H$ can be represented as  $\hv \av^{T} + \bv \hv^{T}$ where $\av, \bv \in \mathbb{R}^{n}.$ Since $H$ is symmetric, it holds that
    \vspace{-1.5mm}
    $$\hv \av^{T} + \bv \hv^{T} = \av \hv^{T} + \hv \bv^{T},$$
    or equivalently
\begin{align}
\hv (\bv -\av)^{T} = (\bv -\av) \hv^{T}.
\label{eq.rankone}
\end{align}
Given the fact in eq. (\ref{eq.rankone}), the two rank-one matrices $\hv (\bv -\av)^{T}$ and $(\bv -\av) \hv^{T}$ will have the same row space and column space, which implies that $\bv-\av = k \hv$, where $k = \frac{(\bv - \av)^{T}\hv}{\|\hv\|_2^2}$.

This implies that $H$ can be written as
$$ H = \hv \av^{T} + \bv \hv^{T} = \hv \av^{T} + \av \hv^{T} + k \hv\hv^{T}, $$
\vspace{-1.5mm}
i.e.
$$ H= (\av + \frac{k}{2}\hv)\hv^{T} + \hv(\av + \frac{k}{2}\hv)^{T}.$$
\vspace{-1.5mm}

Let $\mathbf{c} = (c_1,c_2,\cdots,c_n)^{T} = \av + \frac{k}{2}\hv$.
	Notice that the $i$-th element of the first column of $H$ equals $h^{i-1}c_1 + c_i$.
	By construction, the first column of $H$ is a zero vector. Hence for $i=1$, it holds that $2c_1 = 0$, i.e. $c_1 =0$.
    Thus the $i$-th element of the first column of $H$ equals $c_i$, which implies that $c_2 =  \cdots = c_n=0$, i.e. $\mathbf{c} = 0$. This concludes that $H = 0$.
\end{proof}

Lemma 1 provides a sufficient condition for Theorem 1 to hold.
\begin{Lemma}
	If for any $H\ne 0$ as in eq. (\ref{eq.H}), one has that
	\begin{equation}
		|\trace(PH)| < \|QHQ\|_{\ast},
		\label{eq.ass}
	\end{equation}
	the optimization problem (\ref{eq.opt}) recovers $G_0$ exactly.
\end{Lemma}

\begin{proof}
	Let $V \in \mathbb{R}^{n\times (n-1)}$ be a matrix which satisfies $VV^{T} = Q$ and $V^{T}V = I_{n-1}$.
	The sub-gradients of $\|\cdot\|_{*}$ at $G_0$ are given as the set (see e.g. \cite{c3}):
	\begin{equation}
		\SUB_{\mathbf{h}}
		= \left\{ P + VBV^T \ : \|B\|_2 \le 1\right\}.
	\end{equation}
	By the property of sub-gradient, it holds that for any $H$ as in eq. (\ref{eq.H}),
 	\begin{align}
 		\|G_0+H\|_{*} \geq \|G_0\|_{*} + \langle H, F\rangle, \nonumber
 	\end{align}
 	where $F \in \mathbb{R}^{n\times n}$ is any matrix which belongs to $\SUB_{\mathbf{h}}$.

 	Hence, for any $H$, if there exists a matrix in $\SUB_{\mathbf{h}}$, i.e. a $B$ with $\|B\|_2 \le 1$, such that
	\begin{align} \nonumber
 		\left\langle H, P + VBV^T\right\rangle > 0  \nonumber
	\end{align}
or equivalently
\begin{align}
 \trace (HP) > \langle V^THV,-B\rangle,
\label{eq.ine1}
\end{align}
	then $\|G_0+H\|_{*} > \|G_0\|_{*}$ holds, which implies Theorem 1.

 We are left to find a matrix which satisfies inequality (\ref{eq.ine1}) given the assumption (\ref{eq.ass}). From eq. (\ref{eq.ass}), we have that
 	\begin{equation}\nonumber
		|\trace (HP)| <  \|QHQ^T\|_{*},
 	\end{equation}
	and $Q$ is a projection matrix onto an $n-1$ dimensional subspace, then
	\begin{equation}\nonumber
		\|QHQ^T\|_{*}
		=\|V^THV\|_{*},
	\end{equation}
	which gives that
 	\begin{equation} \nonumber
 		|\trace (HP)|
 		< \|V^THV\|_{*}.
	\end{equation}

	Furthermore, it follows from Proposition 1 that there exists a matrix $B_1$ with $\|B_1\|_2\le1$, such that
	\begin{equation}\|V^THV\|_{*} = \left\langle V^THV, B_1\right\rangle, \nonumber\end{equation}
	therefore it holds that
	\begin{equation}|\trace (HP)| < \left\langle V^THV, B_1\right\rangle.\nonumber \end{equation}
	Hence
	\begin{equation*}
    \trace (HP) > -\left\langle V^THV, B_1\right\rangle = \langle V^THV,-B_1\rangle
    \end{equation*}
	 holds, which gives that the inequality (\ref{eq.ine1}) holds for $B_1$. This concludes the proof.
\end{proof}

Next, we prove that the condition in Lemma 1 will always hold whenever $H \ne 0$. Lemma 2 constructs a matrix $M_0$ which will be used in Lemma 3 to constructs the 'certificate' $M_1$.

%The {\em sketch} for proving Lemma 2 is as follows. First, we construct two matrices, namely the $Q_1$ and the $Q_2$, by considering two related linear equations. Then, we can have four related properties about $Q_1$ and $Q_2$, i.e. {\em Fact 1, Fact 2, Fact 3 and Fact 4}, which will be used in the derivation. Finally, we construct a matrix $M_0$ based on $Q_1$ and $Q_2$.
   	
\begin{Lemma}
	Given the matrices $G_i,P,Q\in\R^{n\times n}$ defined as before, then there exists a matrix $M_0\in\R^{n\times n}$ with $\|M_0\|_2< 1$, such that
	\begin{align}
		\trace(QG_iQM_0 - G_iP) = 0, \ \forall i = 1,2,\dots,n-1.
		\label{eq.tre}
	\end{align}
\end{Lemma}

\begin{proof}
	We will give a construction of such matrix $M_0$.
	Let $r>0$ denote the norm of vector $\hv$, which clearly satisfies $$r^2 = \|\hv\|_2^2 = 1+ h^2 + \dots + h^{2n-2}.$$
	We construct two matrices $Q_1\in\R^{n\times n}$ and $Q_2\in\R^{n\times n}$
    	which satisfy the following two equations:
	\begin{align}
    \label{eq.q1pq2}
		r^2(Q_1 + Q_2)
		= r^2 Q
		= r^2I_n - G_0,
	\end{align}
    and
    \begin{align}
		&r^2(Q_1 - Q_2) =  \label{eq.q1mq2}\\
		&{\scriptsize \nonumber
		    \begin{bmatrix}
			-h^n               &  -h^{n+1}                  & -h^{n+2}              & \cdots               & -h^{2n-2}  &s\\
			-h^{n+1}           &  -h^{n+2}                  & \vdots                &-h^{2n-2}             & s    &-1\\
			-h^{n+2}           &  \vdots                    & -h^{2n-2}             & s     & -1                 &-h \\
			\vdots             &  -h^{2n-2}                 & s      & -1                   & \vdots             &\vdots\\
			-h^{2n-2}          &  s          & -1                    & \vdots               & \vdots             &-h^{n-3}\\
			s   &  -1                        & -h                    & \cdots               & -h^{n-3}           &-h^{n-2}
			\end{bmatrix}},
	\end{align}
 where $s = h+h^3+\dots+h^{2n-3}$.
	
	The matrices $Q_1, Q_2$ satisfy the following properties ( the proofs of them are left to the appendix):

	\begin{itemize}
	\item{\em Fact 1:}
	\begin{equation}
    (Q_1+Q_2)(Q_1+Q_2) = (Q_1-Q_2)(Q_1-Q_2).
    \label{eq.fact1}
    \end{equation}

	\item{\em Fact 2:}
    \begin{equation}
    (Q_1+Q_2)(Q_1-Q_2) = (Q_1-Q_2)(Q_1+Q_2).
    \label{eq.fact2}
    \end{equation}

    \item{\em Fact 3:}
    $$ Q_1Q_2 = Q_2Q_1 = 0.$$

	\item{\em Fact 4:}
	$$Q_1^2 = Q_1,\ \ Q_2^2 = Q_2.$$
	\end{itemize}

Fact 3 and Fact 4 imply that the eigenvalues of matrix $(Q_1-Q_2)$ are in the set $\{  1, -1 , 0\}$, which gives that the spectral norm of $(Q_1-Q_2)$ is 1.

	These properties lead us to consider the following choice of $M_0$
	\begin{equation}
		M_0 = -h^{n}(Q_1-Q_2).
		\label{eq.M}
	\end{equation}
	Now we can prove that matrix $M_0$ satisfies all the equalities given in eq. (\ref{eq.tre}) based on these Facts.
	First, notice that the equalities in eq. (\ref{eq.tre}) are equivalent to the following equalities
	\begin{equation}
		\trace\left(G_i(QM_0Q - P)\right) = 0, \ \forall i = 1,2,\dots,n-1.
		\label{eq.tre1}
	\end{equation}
	The term $QM_0Q - P$ in eq. (\ref{eq.tre1}) can be calculated as follows:
	\begin{align}
	  	&QM_0Q - P\label{eq.reduction}\\
		=\,&-h^n(Q_1+Q_2)(Q_1-Q_2)(Q_1+Q_2) - P \nonumber\\
        =\,& -h^n(Q_1-Q_2)(Q_1+Q_2) - P \nonumber\\
        =\,& -h^n(Q_1-Q_2) - P \nonumber\\
		=\,& M_0-P. \nonumber
    	\end{align}

% In the previous derivations, we have made use of the fact that
%    \begin{align*}
%    (Q_1+Q_2)(Q_1-Q_2) = (Q_1-Q_2),
%    \end{align*}
% which could be verified by expanding the left hand side and using the fact that both $Q_1$ and $Q_2$ are projection matrices.
    So, proving the equalities in eq. (\ref{eq.tre1}) boils down to prove that
    	\begin{equation}
		\trace\left(G_i(M_0-P)\right) = 0, \ \forall i = 1,2,\dots,n-1.
		\label{eq.tre2}
	\end{equation}

	Notice that $M_0$ has the same elements as $P$ in the lower anti-diagonal part, so eq. (\ref{eq.tre2}) holds, which in turn implies  eq. (\ref{eq.tre}).

Finally, notice that $\|M_0\|_2 = |h|^n$, which is less than 1. This gives that $M_0$ is the desired matrix and concludes the proof.
\end{proof}

Based on the constructed $M_0$ in Lemma 2, we can certify that:
\begin{Lemma}
	For any $H\ne 0$, we have that
	\begin{equation}
		|\trace(PH)| < \|QHQ\|_\ast.
		\label{eq.hp}
	\end{equation}
\end{Lemma}
\vspace{+2mm}
\begin{proof}
	We distinguish between two cases, namely
	\begin{equation}
		\trace(PH) < \|QHQ\|_\ast,
		{\label{ineq.lm2}}
	\end{equation}
	and
	\begin{equation}
		-\trace(PH) < \|QHQ\|_\ast.\nonumber
	\end{equation}
	We will give a derivation of eq.  (\ref{ineq.lm2}), the latter inequality follows along the same lines.
    	With the application of Proposition 1, it follows that to prove eq. (\ref{ineq.lm2}) is equivalent to prove that
	\begin{equation}
		\sup_{\|M\|_2\le 1} \trace\left(QHQM - HP\right) > 0. \label{eq.1}
	\end{equation}
	Notice that $H=\sum_{i=1}^{n-1} v_iG_i$, and that by construction of $M_0$ in Lemma 2,
	we have that
	\begin{equation}
		\trace\left(QHQM_0 - HP\right) \\ \nonumber
		= \sum_{i=1}^{n-1} v_i \trace (QG_iQM_0 - G_iP) = 0. \nonumber
	\end{equation}
	Next, observe that $M_0$ is strictly inside the ball $\mathcal{B}_M = \{M: \|M\|_2\le 1, M \in \mathbb{R}^{n \times n}\}$, hence there exists a small value $\delta >0$
	such that
	\begin{equation}
		M_1 = M_0 + \delta( QHQ),	
	\end{equation}
	will also be inside $\mathcal{B}_M$.
	Since $H \neq 0$, it follows that $QHQ \neq 0$ given by Proposition 2,  which implies that
	\begin{equation}
		\trace\left(QHQM_1 - HP\right)
		= \delta\trace{(QHQQHQ)} \\ \nonumber
		= \delta \|QHQ\|_{F}^2\nonumber
	\end{equation}
	is positive. This certifies eq. (\ref{eq.1}) and hence eq. (\ref{ineq.lm2}), which in turn concludes the proof of Lemma 3.
\end{proof}

In conclusion, application of Lemmas 1, 2 and 3 gives Theorem 1.
\end{subsection}
\end{section}

\begin{section}{Discussions}

	The previous sections study rank one Hankel matrix completion problem where the revealed entries follow a deterministic pattern. It is natural to raise the question whether the nuclear norm heuristic will still work when the rank of the Hankel matrix is greater than 1. This is not always the case. We will provide a numerical illustration to show this.

Consider a second order system with the impulse response given by $\{h_1^{i-1}+h_2^{i-1}\}_{i=1}^{\infty}$, where $h_1,h_2 \in \mathbb{R}$ and $-1<h_1,h_2<1$ represent the two poles of the system. Assume that $n=10$, then the matrix completion problem based on the nuclear norm heuristic is given as
	\begin{align}
        \label{eq.multpl}
		\hat{G} = &\argmin_{G\in\mathcal{H}_{10}} \|G\|_{*} \\ \nonumber
		&\mbox{ \ s.t. \ } G(i,j) = G_0(i,j), \forall i+j\le 11. \nonumber
	\end{align}
	Fig. \ref{fig.exp2} displays $\|G_0\|_{*}- \|\hat{G}\|_{*}$ for different choices of $h_1$ and $h_2$, which are chosen from $h_1 = -0.94:0.05:0.94$ and $h_2 = -0.94:0.05:0.94$.
	From this experiment, it becomes clear that  $G_0$ does not always has minimal nuclear norm, and recovery by eq. (\ref{eq.multpl}) will not necessarily succeed. However, it is worthwhile to mention that, in most cases, the nuclear norm heuristic gives the correct recovery.

Hence we conclude the note with the following open questions:
(1) A rigorous characterization of when the nuclear norm heuristic will work in the stable multiple-pole system case is in order. By inspecting the proofs for the single pole case, we can see that the Lemmas 1 and 3 are still applicable. When the matrix $M_0$ in Lemma 2 is constructed, then the 'certificate' $M_1$ can be constructed using the same idea as in Lemma 3. However, construction of $M_0$ becomes more complicated in this situation. (2)When the nuclear norm heuristic doesn't work, see the cases in the previous experiment, it is interesting to find out what kind of additional assumptions could assist the heuristic to work. (3) The studies in this note assume noiseless data, it is not clear how this assumption can be relaxed in the noisy case.

 \begin{figure}[htbp]
\centering
\includegraphics[width = 85mm]{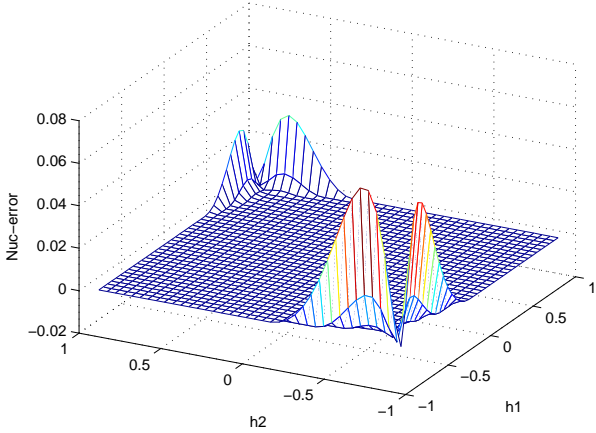}
  \caption{
  	This figure displays $\|G_0\|_{*}$- $\|\hat{G}\|_{*}$ for a range of 2 real poles.
	It is seen that the nuclear norm objective value is not always minimal for the true system $G_0$ for many choices of $h_1$ and $h_2$,
	implying that the heuristic will not always work for such systems.
	Note that the difference equals zero for the case where $h_1=h_2$ as confirmed by Theorem 1. }
\label{fig.exp2}
\end{figure}

\end{section}

\appendix
\begin{section}{Proofs of Fact 1, 2, 3 and 4 in Lemma 2}

Define $\Delta = r^2(Q_1-Q_2)$, and its $k$th column as $\Delta_{k} =\deltabf_{k}+hr^2\ev_{n-k+1},$
in which
\begin{align*}
\deltabf_{k} =\left[ -h^{n+k-1},\cdots, -h^{2n-1},-1,\cdots,-h^{k-2}\right]^{T}.
\end{align*}
%\vspace{-8mm}
\subsection{Proof of Fact 1}
%\vspace{-4mm}
	Notice that to prove (\ref{eq.fact1}) is equivalent to prove
	$$r^2(Q_1+Q_2)r^2(Q_1+Q_2) = \Delta^2.$$
	By definition, one has that
	$$r^2(Q_1+Q_2)r^2(Q_1+Q_2) = (r^2I_n - G_0)^2= r^4 I_n-r^2(\mathbf{h}\mathbf{h}^T).$$
    So, we are left to prove that
    $$\Delta^2 = r^4 I_n-r^2(\mathbf{h}\mathbf{h}^T).$$
	We will calculate out the off-diagonal entries and the on-diagonal entries of $\Delta^2$ separately.
	\begin{itemize}
%    \vspace{-3mm}
	\item{\em The off-diagonal elements.}

		Take for any $1\leq k_1,k_2\le n$ the corresponding columns from $\Delta$ as $\Delta_{k_1}$ and $\Delta_{k_2}$ (assume that $k_1<k_2$ without loss of generality). The $(k_1,k_2)$ and $(k_2,k_1)$ elements of $\Delta^2$ are given by
	\begin{align*}
		\Delta_{k_1}^T\Delta_{k_2} &=(\deltabf_{k_1}+hr^2\ev_{n-k_1+1})^{T}(\deltabf_{k_2}+hr^2\ev_{n-k_2+1}) \\
        &= \deltabf_{k_1}^{T}\deltabf_{k_2} + hr^2\left( \deltabf_{k_1}^{T}\ev_{n-k_2+1} +\ev_{n-k_1+1}^{T}\deltabf_{k_2}\right).
        \end{align*}

        It holds that
        \begin{align*}
        &\delta_{k_1}^{T}\delta_{k_2} = \\
        \vspace{+3.5mm}
		&\quad \, \left[ -h^{n+k_1-1},-h^{n+k_1},\cdots,-h^{2n+k_1-k_2-1}\right]\\
        &\qquad \quad \quad \quad \left[ -h^{n+k_2-1},-h^{n+k_2},\cdots,-h^{2n-1}\right]^{T} \\
        \vspace{+2.5mm}
		 & + \left[ -h^{2n+k_1-k_2},-h^{2n+k_1-k_2+1},\cdots,-h^{2n-1}\right]\\
        &\qquad \qquad \qquad \qquad \, \, \,\left[ -1,-h,\cdots,-h^{k_2-k_1-1}\right]^{T} \\
        \vspace{+2.5mm}
		&  + \left[ -1,-h,\cdots,-h^{k_1-2}\right]\\
        &\qquad \quad \quad \, \, \left[ -h^{k_2-k_1},-h^{k_2-k_1+1},\cdots,-h^{k_2-2}\right]^{T},
	\end{align*}
and also
\begin{align*}
\delta_{k_1}^{T}\ev_{n-k_2+1} +  \ev_{n-k_1+1}^{T}\delta_{k_2}= -h^{2n+k_1-k_2-1}-h^{k_2-k_1-1}.
\end{align*}
	Combing these terms and some algebraic calculations give that
	\begin{align*}
		\Delta_{k_1}^T\Delta_{k_2}&= -(h^{k_2+k_1-2}+h^{k_2+k_1}+\cdots+h^{2n+k_2+k_1-4}) \\
		&= -h^{k_1+k_2-2}r^2,
	\end{align*} which is desired.

\item{\em The diagonal elements.}

For $1\le k\le n$, we need to verify that
\begin{align*}
&\Delta_{k}^{T}\Delta_{k} =\\
&(1+h^2+\cdots+h^{2n-2})^2-(1+h^2+\cdots+h^{2n-2})h^{2k-2},
\end{align*}
in which $\Delta_{k} =\deltabf_{k}+hr^2\ev_{n-k+1}.$

 Notice that
 \begin{align*}
 \Delta_{k}^{T}\Delta_{k}&= \|\deltabf_{k}\|_2^2 + 2hr^2\deltabf_{k}^{T}\ev_{n-k+1} + h^2r^4\\
 & =(h^{2n+2k-2}+h^{2n+2k}+\cdots+h^{4n-4}+h^{4n-2})\\
      &\ \ \ +(1+h^2+\cdots+h^{2k-4}) - 2r^2h^{2n}+ h^2r^4.
 \end{align*}

 Furthermore, it holds that
 \begin{align*}
 h^{4n-2} - 2r^2h^{2n}+ h^2r^4& = (hr^2 - h^{2n-1})^2\\
 & =(h+h^3+\cdots+h^{2n-3})^2.
 \end{align*}

It remains to verify that
	\begin{align}\nonumber
		&(h^{2n+2k-2}+h^{2n+2k}+\cdots+h^{4n-4})\\ \nonumber
      &\ \ \ +(h+h^3+\cdots+h^{2n-3})^2 +(1+h^2+\cdots+h^{2k-4}) \\ \nonumber
		&=  (1+h^2+\cdots+h^{2n-2})^2-(1+h^2+\cdots+h^{2n-2})h^{2k-2}, \nonumber
	\end{align}
	which is equivalent to verify
	\begin{align*}
		(1+h^2+\cdots+h^{4n-4})+(h+h^3+\cdots+h^{2n-3})^2 \\
= (1+h^2+\cdots+h^{2n-2})^2.
	\end{align*}
	This follows from the reasonings as below
	\begin{align}\nonumber
		&{\tiny \Leftrightarrow \frac{(1-h^2)(1-h^{4n-2})+h^2(1-h^{2n-2})^2}{(1-h^2)^2} = \frac{(1-h^{2n})^2}{(1-h^2)^2}} \nonumber \\
		&\Leftrightarrow 1 + h^{4n} - 2h^{2n} = (1 - h^{2n})^2.\nonumber
	\end{align}
\end{itemize}

%\vspace{-3mm}
\subsection{Proof of Fact 2}
%\vspace{-3mm}
We firstly prove that the vector $\mathbf{h}$ lies in the null space of $\Delta$. Take for any $1\leq k \le n$ the corresponding column from the matrix $\Delta$, we have that
        \begin{align*}
        		\Delta_{k}^T\mathbf{h} &= \deltabf_k^{T}\hv + hr^2 \ev_{n-k+1}^{T}\hv \\
        & = -(h^{n+k-1}+h^{n+k+1}+\cdots+h^{3n-k-1}) \\
        & \quad \, - (h^{n-k+1}+h^{n-k+3}+\cdots+h^{n+k-3})\\
        & \quad \, + (h^{n-k+1}+h^{n-k+3}+\cdots+h^{3n-k-1})\\
        & =  0.
       \end{align*}
 Hence  $(Q_1-Q_2)\mathbf{h} = \mathbf{h}^{T}(Q_1-Q_2) = 0$ holds. Therefor it holds that $(Q_1-Q_2)P = P(Q_1-Q_2) = 0$. With this, the Fact 2 can be concluded by the following
    \begin{align*}
		&(Q_1+Q_2)(Q_1-Q_2) = (Q_1-Q_2)(Q_1+Q_2) \\
		&\Leftrightarrow (I_n - P)(Q_1-Q_2) = (Q_1-Q_2)(I_n - P) \\
        &\Leftrightarrow  P(Q_1-Q_2) = (Q_1-Q_2)P.
	\end{align*}

%\vspace{-3mm}
\subsection{Proof of Fact 3}
%\vspace{-3mm}
	From eq. (\ref{eq.fact1}), we have that
    \begin{align*}
    Q_2Q_1 + Q_1Q_2 = -Q_2Q_1-Q_1Q_2.
    \end{align*}
    From eq. (\ref{eq.fact2}), we have that
     \begin{align*}
    Q_2Q_1 - Q_1Q_2 = -Q_2Q_1+Q_1Q_2.
    \end{align*}
    Hence we can conclude that $Q_1Q_2=Q_2Q_1=0$.

%\vspace{-3mm}
\subsection{Proof of Fact 4}
%\vspace{-3mm}
	As shown in the proof of Fact 2, we have that $(Q_1-Q_2)P = 0.$
    Together with $(Q_1+Q_2)P = 0$, we have that $Q_1P=0$ and $Q_2P= 0$. Hence, the Fact 4 can be concluded as follows
    \begin{align*}\nonumber
    &Q_1 = Q_1(Q_1+Q_2+P) = Q_1^2 + Q_1Q_2 + Q_1P = Q_1^2, \\
    &Q_2 = Q_2(Q_1+Q_2+P) = Q_1Q_2 + Q_2^2 + Q_2P = Q_2^2.
    \end{align*}
\end{section}


\begin{thebibliography}{99}

\bibitem{c8}
L. Vandenberghe, Convex optimization techniques in system identification, {\em 16th IFAC Symposium on System Identification}, Brussels, Belgium, July 2012.
\bibitem{c1}
I. Markovsky, How effective is the nuclear norm heuristic in solving data approximation problems?, {\em 16th IFAC Symposium on System Identification}, Brussels, Belgium, July 2012.
\bibitem{c4}
M. Fazel, H. Hindi, and S. Boyd, Log-det heuristic for matrix rank minimization with applications to hankel and euclidean distance matrices, {\em Proceeding of American Control Conference}, Denver, Colorado, June 2003.
\bibitem{c5}
M. Fazel, H. Hindi, and S. Boyd, A rank minimization heuristic with application to minimum order system approximation, {\em Proceeding of American Control Conference}, Arlington, Virginia, June 2001.
\bibitem{c6}
E. J. Cand\'{e}s and B. Recht, Exact matrix completion via convex optimization, {\em Foundation of Computational Mathematics}, Vol 9, pp. 717-772, 2009.
\bibitem{c7}
D. Gross, Recovering low-rank matrices from few coefficients in any basis, {\em IEEE Transaction on Information Theory}, Vol 57, pp. 1548- 1566, 2011.
\bibitem{c3}
B. Recht, M. Fazel and P. A. Parrilo, Guaranteed minimum rank solutions to linear matrix equations via nuclear norm minimization, {\em SIAM Review}, Vol 52, no 3, pp. 471-501, 2010.
\bibitem{c10}
 Z. Liu, L.  Vandenberghe, Semidefinite programming methods for system realization and identification, {\em Proceedings of the 48th IEEE Conference on Decision}, Shanghai, China, 2009.
% \bibitem{c11}
%A. Tether, Construction of minimal linear state-variable models from finite input-output data, {\em IEEE Transactions on  Automatic Control}, Vol 15, pp. 427-436, 1970.
\end{thebibliography}
\end{document}